\documentclass[letterpaper, 9 pt, conference]{ieeeconf}  % Comment this line out if you need a4paper

\usepackage{mathrsfs}
\DeclareMathAlphabet{\mcz}{OT1}{pzc}{m}{it}
\usepackage{graphics} % for pdf, bitmapped graphics files
\usepackage{graphicx}
\usepackage{chngcntr}
\usepackage{apptools}
\usepackage{multicol}
\usepackage{epsfig} % for postscript graphics files
\usepackage{float}
\usepackage{times} % assumes new font selection scheme installed
\usepackage{stdmath}
\usepackage{amsmath}  % assumes amsmath package installed
\usepackage{esint}
\usepackage{amsfonts}
\newtheorem{theorem}{Theorem}

\newtheorem{lemma}{Lemma}

\newtheorem{corollary}{Corollary}
\usepackage{graphicx}
\usepackage{caption}                                                                                                                                                                                                                                                                                                                                                                                                                                                                                                                                                                                                                                                                                                                                                                                                                                                                                                                                                                                                                                                                                                                                                                                                                                                                                                                                                                                                                                                                                                                                                                                                                                                                                                                                                                                                                                                                                                                                                                                                                                                                                                                                                                                                                                                                                                                                                                                                                                                                                                                                                                                                                                                                                                                                                                                                                                                                             
\usepackage{color}
\usepackage{subfig}
\makeatother

\newcommand{\pft}{\partial_t}
\newcommand{\pfx}{\partial_x}
\newcommand{\pfy}{\partial_y}

\newcommand{\hlf}{\frac{1}{2}}

\newcommand{\mcl}[1]{\mathcal{#1}}

\newcommand{\wb}{\bar{w}_\alpha}
\newcommand{\w}{w_\alpha}
\allowdisplaybreaks

\IEEEoverridecommandlockouts                              % This command is only needed if 
\title{\LARGE \bf
A Semi-Definite Programming Approach to Stability Analysis of Linear Partial Differential Equations
}

\author{Aditya Gahlawat$^{1}$ and Giorgio Valmorbida$^{2}$% <-this % stops a space
\thanks{$^{1}$Aditya Gahlawat is with the Department of Mechanical, Materials and Aerospace Engineering (MMAE),
        Illinois Institute of Technology, Chicago, IL-60616, USA
        {\tt\small agahlawa@hawk.iit.edu}}%
\thanks{$^{2}$Giorgio Valmorbida is an associate professor at  Laboratoire des Signaux et Syst\`{e}mes, CentraleSup\'{e}lec, CNRS, Univ. Paris-Sud, Universit\'{e} Paris-Saclay, 3 Rue Joliot Curie, Gif-sur-Yvette 91192, France
        {\tt\small giorgio.valmorbida@l2s.centralesupelec.fr}}%
}

\begin{document}

\maketitle
\thispagestyle{empty}
\pagestyle{empty}

%%%%%%%%%%%%%%%%%%%%%%%%%%%%%%%%%%%%%%%%%%%%%%%%%%%%%%%%%%%%%%%%%%%%%%%%%%%%%%%%
\begin{abstract}

We consider the stability analysis of a large class of linear 1-D PDEs with polynomial data. This class of PDEs contains, as examples, parabolic and hyperbolic PDEs with spatially varying coefficients, and systems of in-domain/boundary coupled PDEs. Our approach is Lyapunov based which allows us to reduce the stability problem to verification of integral inequalities on the subspaces of Hilbert spaces. Then, using fundamental theorem of calculus and Green's theorem, we construct a polynomial optimization problem to verify the integral inequalities. Constraining the solution of the polynomial optimization problem to belong to the set of sum-of-squares polynomials subject to affine constraints allows us to use semi-definite programming to algorithmically construct Lyapunov certificates of stability for the systems under consideration. We also provide numerical results of the application of the proposed method on different types of PDEs.  
\end{abstract}

%%%%%%%%%%%%%%%%%%%%%%%%%%%%%%%%%%%%%%%%%%%%%%%%%%%%%%%%%%%%%%%%%%%%%%%%%%%%%%%%
\section{INTRODUCTION}\label{sec:intro}

Temporal evolution of processes involving spatially distributed physical quantities requires Partial Differential Equations (PDEs) to produce accurate models for analysis and control~\cite{witrant2007control},~\cite{bovskovic2003stabilization}. Analysis of PDE systems brings forth technical challenges relative to Ordinary Differential Equations (ODEs) due to the infinite dimensional nature of the state-space. Owing to the maturity of research on ODEs, a traditional method of analysis is to reduce PDE systems, using spatial or spectral methods, to a system of ODEs which can then be analyzed with relative ease~\cite{el2003analysis},~\cite{goulart2012global}. 

Recently, various \textit{direct methods}, i.e., methods without finite-dimensional approximation, have been developed for analysis and control of PDEs. For controller synthesis, backstepping is an oft used method wherein the problem of stabilizing boundary feedback design is reduced to a hyperbolic PDE whose solution can be obtained either analytically or numerically~\cite{krstic2008boundary}. Similarly, one may use Lyapunov's second method for infinite-dimensional systems to establish stability~\cite{datko1970extending},~\cite{movchan1959direct}. Lyapunov's method requires that the structure of a Lyapunov Functional (LF) be chosen a priori. In this paper we choose quadratic LF candidates of the form used in our previous work in~\cite{gahlawat2016convex}. The reason for this choice, as we established in~\cite{gahlawat2016convex}, is that such LF candidates are successful in providing stability certificates for parabolic PDEs. Additionally, also in~\cite{gahlawat2016convex}, we showed that such LFs are certificates of stability for parabolic PDEs actuated by backstepping boundary feedback. 

The LF candidates we consider lead to Lyapunov inequalities defined by integral expressions containing quadratic terms in the state defined on one-dimensional (1-D) and two-dimensional (2-D) domains and their respective boundaries. Moreover, the verification of these inequalities has to be performed on the spaces of functions defined by the boundary conditions of the PDE whose stability we wish to establish. The interplay between the 2-D domain/1-D domain and the boundary values have to be thus accounted for with the Fundamental Theorem of Calculus (FTC) and Green's theorem, which are used in a similar way to~\cite{valmorbida2015convex}.

\subsection{Contribution} The paper presents a convex optimization based method for the verification of integral inequalities on (subspaces of) Hilbert spaces and its application to stability analysis of PDE systems with polynomial data. The proposed method contains previous approaches by the authors, see, for e.g.~\cite{gahlawat2016convex},~\cite{valmorbida2016stability}, thus encompassing a larger class of PDEs.

In order to accomplish our goals  we first reduce the problem of stability analysis to the verification of two integral inequalities on (subspaces of) Hilbert spaces. To verify these inequalities, we use FTC and Green's theorem to construct polynomial equations whose solution, if exists, verifies the Lyapunov integral inequalities. Finally, we show that this polynomial equation may be solved by searching for Sum-of-Squares (SOS) polynomials under affine constraints thus rendering the problem of searching for the solution as a Semi-Definite Program (SDP)~\cite[Chapter~3]{blekherman2012semidefinite},~\cite{vandenberghe1996semidefinite}, a convex optimization problem.

\subsection{Outline}

The paper is organized as follows. In Section~\ref{sec:pro_statement} we present the problem statement and formulate the stability analysis problem as the problem of verification of two integral inequalities. In Section~\ref{sec:non_neg_ineq} we use positive semi-definite polynomial matrices to characterize a class of non-negative/positive integral inequalities on Hilbert spaces. In Section~\ref{sec:slack} we use FTC and Green's theorem to construct integral terms which relate the domains on which the LF candidates are defined to their boundaries, which we call slack integrals. Finally, in Section~\ref{sec:main} we present the main result and numerical experiments. 

%%%%%%%%%%%%%%%%%%%%%%%%%%%%%%%%%%%%%%%%%%%%%%%%%%%%%%%%%%%%%%%%%%%%%%%%%%%%%%%%
\subsection{Notation}\label{sec:notation}

We denote by $\overline{\Omega}$ the set $\{(x,y) \in \R^2:0\leq y \leq x \leq 1\}$, $\underline{\Omega}=\{(x,y) \in \R^2:0\leq x \leq y \leq 1\}$ and $\Delta = [0,1] \times [0,1]$. In the following definitions $\alpha,\beta \in \N$. For $w:[0,1] \rightarrow \R^\beta$, $w \in \mcl{C}^\alpha([0,1])$, we denote by $\pfx^i w(x)$ the $i$-th derivative of $w$ and
\begin{align*}
w_\alpha(x) = & \bmat{w(x)^T, & \pfx w(x)^T, & \cdots, & \pfx^\alpha w(x)^T}^T,\\
w_\alpha^b =& \bmat{w_{\alpha-1}(1)^T, & w_{\alpha-1}(0)^T}^T,\\
\wb(x) =  &\bmat{w_\alpha(x)^T, & \left(w_\alpha^b\right)^T}^T.
\end{align*} Thus, $w_\alpha:[0,1] \rightarrow \R^{\beta(\alpha+1)}$, $w_\alpha^b \in \R^{2\beta\alpha}$ and $\wb:[0,1] \rightarrow \R^{\beta(3\alpha+1)}$. We denote by $N_\partial,N \in \N^{\beta\alpha \times \beta(\alpha+1)}$ and $N_0,N_1 \in \N^{\beta\alpha \times 2 \beta\alpha}$ the matrices satisfying 
\begin{alignat*}{2}
&\pfx w_{\alpha-1}(x)=N_\partial w_{\alpha}(x), \quad &&w_{\alpha-1}(x)=N w_{\alpha}(x), \\
&w_{\alpha-1}(0)=N_0 w_\alpha^b,\quad &&w_{\alpha-1}(1)=N_1 w_\alpha^b.
\end{alignat*} 

We denote 
\begin{align*}
\mcl{H}^\alpha\left([0,1];\R^\beta \right)=\Big\{ w:[0,1] &\rightarrow \R^\beta~:~w, \pfx w, \dots, \pfx^{\alpha-1}w\text{ are  }  \\
& \text{absolutely continuous on }[0,1]~\text{and}\\
& \int_0^1 \left(\pfx^{\alpha}w(x)\right)^T \left(\pfx^{\alpha}w(x)\right) dx<\infty   \Big\}.
\end{align*} We also write $\mcl{L}_2\left([0,1];\R^\beta \right) = \mcl{H}^0([0,1])$. The space $\mcl{H}^\alpha\left([0,1];\R^\beta \right)$ is endowed with the standard inner product and norm
\[\norm{w}_{\mcl{H}^\alpha}= \sqrt{\ip{w_\alpha}{w_\alpha}_{\mcl{H}^\alpha}} =\sqrt{\int_0^1 w_\alpha(x)^T w_\alpha(x)dx},\] and the space $\mcl{L}_2\left([0,1];\R^\beta \right)$ has the norm and inner product $\norm{\cdot}_{\mcl{L}_2} = \norm{\cdot}_{\mcl{H}^0}$ and $\ip{\cdot}{\cdot}_{\mcl{L}_2}=\ip{\cdot}{\cdot}_{\mcl{H}^0}$, respectively.

For any $m,n \in \N$, we denote by $0_{m,n}$ the matrix of zeros of dimensions $m$-by-$n$ and $0_m$ when $m=n$. Similarly, we denote by $I_n$ the identity matrix of dimensions $n$-by-$n$. The set of symmetric matrices of dimension $n$-by$n$ is donated by $\S^n$.
For any square matrix $Q$, we denote  $He(Q)=\hlf\left(Q+Q^T\right)$.

We denote by $\mcl{S}^{n}[x]$ and $\mcl{S}^{n}[(x,y)]$ the sets of symmetric polynomial matrices of size $n$-by-$n$ in variables $x$, and $x$ and $y$, respectively. Similarly, we denote by $\mcl{R}^{m \times n}[x]$ and $\mcl{R}^{m \times n}[(x,y)]$ the sets of real polynomial matrices of size $m$-by-$n$ in variables $x$ and $y$. We denote by $\Sigma^n[x]$ the set of Sum-of-Squares (SOS) polynomials which, for all $x \in [0,1]$, belong to $\S^n$. polynomials in the variable $x$. Note that, by definition, an SOS polynomial is non-negative for all $x \in \R$~\cite[Chapter~3]{blekherman2012semidefinite}. We say that a polynomial $S \in \mcl{S}^n[x]$ is positive semi-definite on $[0,1]$ if $S(x) \succeq 0$, for all $x \in [0,1]$. Given $\alpha,\beta,d \in \N$, we define
\begin{equation}\label{eqn:Y}
Y_{q(\alpha,\beta,d)}(x,y)=I_{\beta(\alpha+1)}\otimes z(x,y) \in \R^{q(\alpha,\beta,d) \times \beta(\alpha+1)},
\end{equation} where $q(\alpha,\beta,d)=\hlf \beta(\alpha+1)(d+2)(d+1)$, $z(x,y) \in \R^{\hlf (d+2)(d+1)}$ is the vector of monomials in $x$ and $y$ up to degree $d$ and $\otimes$ denotes the Kronecker product. For example, for $d=2$, we have
\[z(x,y)=\bmat{1 & x & y & x^2 & xy & y^2}^T.\]

For any bivariate function $K:\overline{\Omega} \rightarrow \R^{n \times n}$, we define the linear map
\[\Gamma \left[K \right] = \begin{cases} K(x,y), & \quad x \geq y   \\ 
K(y,x)^T, & \quad y > x \end{cases},  \] thus satisfying, for any $v:[0,1]\rightarrow \R^n$,
\begin{align*}
\int_\Delta v(x)^T \Gamma\left[K\right] v(y) d \Delta =& \int_0^1 \int_0^x v(x)^T K(x,y) v(y) dy dx \\
& + \int_0^1 \int_x^1 v(x)^T K(y,x)^T v(y) dy dx.
\end{align*}

%%%%%%%%%%%%%%%%%%%%%%%%%%%%%%%%%%%%%%%%%%%%%%%%%%%%%%%%%%%%%%%%%%%%%%%%%%%%%%%%
\section{PROBLEM STATEMENT}\label{sec:pro_statement}

Consider the following class of linear PDEs 
\begin{subequations}\label{eqn:PDE_complete}
\begin{align}
&\pft w(x,t) =  A(x)\w(x,t)  , \label{eqn:PDE}
\end{align} where $A \in \mcl{R}^{\beta \times \beta(\alpha+1)}[x]$. The state $w:[0,1] \times [0,\infty) \rightarrow \R^\beta$ belongs to the set $\mcl{B}$ of functions satisfying the boundary conditions, defined as
\begin{align}
\mcl{B}=&\Bigl\{ w \in \mcl{H}^\alpha\left([0,1];\R^\beta \right)~ :~    F w_\alpha^b =0_{\beta\alpha,1}  \Bigr\}, \label{eqn:BC}
\end{align} where $F \in \R^{\beta\alpha \times 2 \beta\alpha}$.
\end{subequations}
We now provide a few examples of PDEs which can be cast in the form of~\eqref{eqn:PDE_complete}. Each of the following PDEs is parameterized by a positive scalar $\lambda$. 

\subsubsection{Example~$1$}\label{exmp1:para} Consider the following \textbf{parabolic PDE with distributed coefficients}
\begin{subequations}\label{eqn:exmp:parabolic}
\begin{align}
\pft v(x,t) = & \left(x^2 + 1  \right)\pfx^2 v(x,t)+ 0.5 x \pfx v(x,t) + \lambda v(x,t), \\
 v(0,t)=&0, \quad \pfx v(1,t)=0.
\end{align}
\end{subequations}  This PDE may be set in the form of~\eqref{eqn:PDE_complete} with $\alpha=2,\beta=1$ and
\begin{align*}
w(x,t)=&v(x,t),\\
A(x)=&\bmat{\lambda & 0.5x & x^2+1},\quad F = \bmat{0 & 0 & 1 & 0 \\ 0 & 1 & 0 & 0}.
\end{align*} 

\subsubsection{Example~$2$}\label{exmp2:coupled_hyperbolic} Now consider the following \textbf{system of hyperbolic PDEs coupled in-domain and at the boundaries}
\begin{subequations}\label{eqn:exmp:hyperbolic} 
\begin{align}
&\pft v_1(x,t)= \left(\lambda - 1\right)\left(x^2+1  \right)\pfx v_1(x,t)+(x-3)v_2(x,t),\\
&\pft v_2(x,t)= \left(x+1  \right)\pfx v_2(x,t),\\
&v_1(0,t)-3v_2(0,t)=0, \quad v_2(1,t)=0.
\end{align}
\end{subequations}
This PDE may be set in the form of~\eqref{eqn:PDE_complete} with $\alpha=1,\beta=2$ and
\begin{align*}
w(x,t)=&\bmat{v_1(x,t) & v_2(x,t)}^T,\\
A(x)=&\bmat{0 & x-3 & (\lambda-1)\left(x^2+1 \right) & 0 \\ 0 & 0 & 0 & x+1},\\
 F = & \bmat{0 & 0 & 1 & -3 \\ 0 & 1 & 0 & 0}.
\end{align*}

\subsubsection{Example~$3$}\label{exmp3:beam} Finally, consider the following \textbf{Euler-Bernoulli beam model} 
\begin{subequations}\label{eqn:exmp:beam}
\begin{align}
&\pft^2 v(x,t) + \pfx^4 v(x,t)=0,\\
& \pfx^2 v(0,t)-\frac{1}{1-\lambda}\partial_{xt}v(0,t)=0, \\
& \pfx^3 v(0,t)+(1-\lambda)\pft v(0,t)=0, \quad \pfx^2 v(1,t)=0,\\
&v(1,t)=0.
\end{align}
\end{subequations}
  We may re-write~\eqref{eqn:exmp:beam} as
\begin{align*}
&\pft \bmat{\pfx^2 v(x,t) \\ \pft v(x,t)}=\bmat{ \pft \pfx^2  v(x,t) \\ -\pfx^4 v(x,t)},\\
& \pfx^2 v(0,t)-\frac{1}{1-\lambda}\partial_{xt}v(0,t)=0, \\
& \pfx^3 v(0,t)+(1-\lambda)\pft v(0,t)=0, \quad \pfx^2 v(1,t)=0,\\
&\pft v(1,t)=0.
\end{align*} With this representation, we may write~\eqref{eqn:exmp:beam} in the form of~\eqref{eqn:PDE_complete} with $\alpha=2,\beta=2$ and
\begin{align*}
w(x,t)=&\bmat{\pfx^2 v(x,t) & \pft v(x,t)}^T,\\
A(x)=&\bmat{0 & 0 & 0 & 0 & 0 & 1 \\ 0 & 0 & 0 & 0 & -1 & 0},\\
 F = &\bmat{0 & 0 & 0 & 0 & 1 & 0 & 0 & -\frac{1}{1-\lambda} \\
              0 & 0 & 0 & 0 & 0 & 1-\lambda & 1 & 0\\
              1 & 0 & 0 & 0 & 0 & 0 & 0 & 0\\
              0 & 1 & 0 & 0 & 0 & 0 & 0 & 0}.
\end{align*}

We wish to establish the stability of~\eqref{eqn:PDE_complete} by constructing Lyapunov functional (LF) certificates of exponential stability. In particular, we wish to construct LFs of the form
\begin{align}
\mcz{V}(w) = & \hlf \int_0^1 w(x,t)^T T_b(x) w(x,t)dx \notag \\
&\label{eqn:V_int} + \hlf \int_\Delta w(x,t)^T \Gamma \left[ \bar{T}\right] w(y,t) d\Delta,
\end{align} where $T_b \in \mcl{S}^{\beta}[x]$ and $\bar{T} \in \mcl{R}^{\beta \times \beta}[(x,y)]$. 
As stated in the Introduction, parabolic PDE systems with boundary backstepping based control laws admit LF certificates of stability which have the same structure as~\eqref{eqn:V_int}~\cite{gahlawat2016convex}. 
The numerical results in~\cite{gahlawat2016convex} also indicate that such LFs are not conservative for spatially distributed scalar parabolic PDEs. 

Let us now formulate conditions for the exponential stability of~\eqref{eqn:PDE_complete}. We begin by computing the time derivative of~\eqref{eqn:V_int} along the trajectories of~\eqref{eqn:PDE_complete}. Consider a scalar $\delta >0$ and define
\begin{equation}\label{eqn:Vdot}
\mcz{V}_d(w) = -\pft \mcz{V}(w)-2\delta \mcz{V}(w),
\end{equation} which can be put in the form (see the Appendix)
\begin{align}
\mcz{V}_d(w)=&\label{eqn:Vd_int} \int_0^1 \wb(x)^T U_b(x) \wb(x)dx \notag \\
&+ \int_\Delta \w(x)^T \Gamma \left[\bar{U} \right] \w(y) d\Delta, \quad w \in \mcl{H}^\alpha\left([0,1];\R^\beta \right),
\end{align} where
\begin{align*}
U_b(x) = & -He \left( \bmat{T_b(x)A(x) & 0_{\beta,2\beta\alpha} \\ 0_{3\beta\alpha,\beta(\alpha+1)} & 0_{3\beta\alpha,2\beta\alpha}}\right) \\
&  -He \left( \bmat{\delta T_b(x) & 0_{\beta,3\beta\alpha} \\ 0_{3\beta\alpha,\beta} & 0_{3\beta\alpha}}\right),\\
\bar{U}(x,y) = &  -\hlf \left(\bmat{ \bar{T}(x,y)^T A(x) \\ 0_{\beta\alpha,\beta(\alpha+1})}^T +
\bmat{\bar{T}(x,y) A(y) \\ 0_{\beta\alpha,\beta(\alpha+1)}} \right)\\
&   -\hlf 
\bmat{2\delta \bar{T}(x,y) & 0_{\beta,\beta\alpha} \\ 0_{\beta\alpha,\beta} & 0_{\beta\alpha}}.
\end{align*}

%%%%%%%%%%%%%%%%%%%%%%%%%%%%%%%
The following theorem casts the verification of exponential stability of~\eqref{eqn:PDE_complete} as the verification of two integral inequalities.
\begin{theorem}\label{thm:analysis}
Given the PDE~\eqref{eqn:PDE_complete}, suppose there exist positive scalars $\mu,\delta$ and polynomial matrices $T_b \in \mcl{S}^{\beta}[x]$ and $\bar{T} \in \mcl{R}^{\beta \times \beta}[(x,y)]$ such that
\begin{alignat}{2}
&\label{eqn:stab:V_condition}\mcz{V}(w) \geq \mu \norm{w}^2_{\mcl{L}_2}, \quad &&\forall w \in \mcl{L}_2\left([0,1];\R^\beta \right),\\ 
&\label{eqn:stab:Vd_condition}\mcz{V}_d(w) \geq 0, \quad &&\forall w \in \mcl{B} \subset \mcl{H}^\alpha\left([0,1];\R^\beta \right),
\end{alignat} where $\mcz{V}(w)$ and $\mcz{V}_d(w)$ are defined in~\eqref{eqn:V_int} and~\eqref{eqn:Vd_int}, respectively.

Then, there exists a positive scalar $\kappa$ such that the solution of~\eqref{eqn:PDE_complete} satisfies
\begin{equation}\label{eqn:stab:conclusion}
\norm{w(\cdot,t)}_{\mcl{L}_2} \leq \kappa e^{- \delta t} \norm{w(\cdot,0)}_{\mcl{L}_2}, \quad \forall t \geq 0.
\end{equation} 
\end{theorem}
%%%%%%%%%%%%%%%%%%%%%%%%%%%
\begin{proof}
Let us choose the LF candidate $\mcz{V}(w)$. Since the condition in~\eqref{eqn:stab:V_condition} holds, we have that there exists a positive scalar $\theta$ such that
\begin{equation}\label{eqn:stab:V}
\mu \norm{w(\cdot,t)}_{\mcl{L}_2}^2 \leq \mcz{V}(w) \leq \theta \norm{w(\cdot,t)}_{\mcl{L}_2}^2, \quad \forall t \geq 0,
\end{equation} where the upper bound is a consequence of $T_b$ and $\bar{T}$ being polynomial matrices defined on bounded domains.

Now, for this LF candidate, we have from~\eqref{eqn:Vdot} that
\begin{align*}
\mcz{V}_d(w)=&-\pft \mcz{V}(w)-2\delta \mcz{V}(w),
\end{align*} and since~\eqref{eqn:stab:Vd_condition} holds for all $w \in \mcl{B}$ and thus for $w$ that solve~\eqref{eqn:PDE_complete}, we have that
\begin{equation}\label{eqn:stab:Vdot}
-\pft \mcz{V}(w)-2\delta \mcz{V}(w) \geq 0, \quad \forall t \geq 0.
\end{equation} Integrating this expression in time produces
\[\mcz{V}(w) \leq e^{-2\delta t} \mcz{V}(w(0)).\] Using~\eqref{eqn:stab:V} produces
\[\mu \norm{w(\cdot,t)}_{\mcl{L}_2}^2 \leq e^{-2 \delta t} \theta \norm{w(\cdot,0)}_{\mcl{L}_2}^2.\]
Then, taking the square root we conclude that~\eqref{eqn:stab:conclusion} holds with $\kappa = \sqrt{\theta/\mu}$.
\end{proof}

We have reduced the problem of stability analysis of~\eqref{eqn:PDE_complete} to the verification of the integral inequalities in~\eqref{eqn:stab:V_condition} and~\eqref{eqn:stab:Vd_condition}. Thus, the remainder of the work considers the following problem:

\textbf{Problem:} Verify if $\mcz{V}(w)$ is strictly positive, i.e.,~\eqref{eqn:stab:V_condition} is verified on $\mcl{L}_2\left([0,1];\R^\beta \right)$, and $\mcz{V}_d(w)$ is non-negative on $\mcl{B} \subset \mcl{H}^\alpha\left([0,1];\R^\beta \right)$.

We verify~\eqref{eqn:stab:V_condition} by generalizing the methods proposed in~\cite{gahlawat2016convex} and~\cite{valmorbida2016stability}. Namely, we test the positive semi-definiteness  of a polynomial matrix associated to $T_b$ and $\bar{T}$. Since~\eqref{eqn:stab:Vd_condition} requires $\mcz{V}_d(w) \geq 0$ only on the subset $\mcl{B}$ of $\mcl{H}^\alpha\left([0,1];\R^\beta \right)$, it calls for a different formulation than the one adopted to verify~\eqref{eqn:stab:V_condition}. In this case we first construct a set of integral terms $\mcz{S}(w)$ which have the same structure as $\mcz{V}_d(w)$ and satisfy $\mcz{S}(w)=0$, for all $w \in \mcl{B}$. We refer to such expressions as \textit{slack integrals}. We then check for the positive semi-definiteness of a polynomial matrix associated to $\mcz{V}_d(w) + \mcz{S}(w)$ which guarantees $\mcz{V}_d(w) + \mcz{S}(w) \geq 0$, for all $w \in \mcl{H}^\alpha\left([0,1];\R^\beta \right)$. 

The construction of the desired slack integrals is presented in Section~\ref{sec:slack}. These integrals generalize the results presented in~\cite{peet2014lmi},~\cite{peet2009positive} for time-delay systems and in~\cite{valmorbida2016stability} for PDE analysis using~\eqref{eqn:V_int} with $\bar{T}=0$.

%%%%%%%%%%%%%%%%%%%%%%%%%%%%%%%%%%%%%%%%%%%%%%%%%%%%%%%%%%%%%%%%%%%%%%%%%%%%%%%%%%%%%%%%
\section{POSITIVE/NON-NEGATIVE INTEGRAL INEQUALITIES ON HILBERT SPACES}\label{sec:non_neg_ineq}

In this section we construct a set of positive/non-negative integral inequalities which are parameterized by Positive Semi-Definite (PSD) polynomial matrices. The results provided are a generalization of~\cite[Theorem~1]{gahlawat2016convex}. Throughout this section we will use the polynomial matrix $Y_{q(\alpha,\beta,d)}(x,y)$ defined in~\eqref{eqn:Y}.

We begin by constructing non-negative integral inequalities on $\mcl{H}^\alpha\left([0,1];\R^\beta \right)$, given $\alpha,d \in \N$, that have the same form as $\mcz{V}_d(w)$ in~\eqref{eqn:Vd_int}. Let us define
\begin{align}
\mcz{R}(w) \hspace{-1mm}= \hspace{-1mm}&\label{eqn:Rb_Rbar} \int_0^1 \wb(x)^T R_b(x) \wb(x)dx \hspace{-1mm} + \hspace{-1mm} \int_\Delta \w(x)^T \Gamma \left[\bar{R} \right] \w(y)d\Delta,
\end{align} where $R_b \in \mcl{S}^{\beta(3\alpha+1)}[x]$ and
\begin{align}
\bar{R}(x,y)=& R_{12}(x) Y_{q(\alpha,\beta,d)}(x,y) + Y_{q(\alpha,\beta,d)}(y,x)^T R_{13}(y)^T \notag  \\
& + \int_0^y Y_{q(\alpha,\beta,d)}(z,x)^T R_{33} Y_{q(\alpha,\beta,d)}(z,y)dz \notag \\
& + \int_y^x Y_{q(\alpha,\beta,d)}(z,x)^T R_{23}^T Y_{q(\alpha,\beta,d)}(z,y)dz \notag \\
&\label{eqn:Rbar}  +  \int_x^1 Y_{q(\alpha,\beta,d)}(z,x)^T R_{22} Y_{q(\alpha,\beta,d)}(z,y)dz.
\end{align} for some (polynomial) matrices $R_{12},R_{13} \in \mcl{R}^{\beta(\alpha+1) \times q(\alpha,\beta,d)}[x]$, $R_{22},R_{33} \in \S^{q(\alpha,\beta,d)}$ and $R_{23} \in \R^{q(\alpha,\beta,d) \times q(\alpha,\beta,d)}$, which also define
\begin{align}
&R(x) \hspace{-1mm}= \hspace{-1mm} \left[ \begin{array}{c|c|c}
R_{b}(x) & \bmat{ R_{12}(x) \\ 0_{2\beta\alpha,q(\alpha,\beta,d)}}  & \bmat{ R_{13}(x) \\ 0_{2\beta\alpha,q(\alpha,\beta,d)} }\\ \hline
 \bmat{ R_{12}(x) \\ 0_{2\beta\alpha,q(\alpha,\beta,d)}}^T & R_{22} & R_{23} \\ \hline
 \bmat{ R_{13}(x) \\ 0_{2\beta\alpha,q(\alpha,\beta,d)} }^T & R_{23}^T & R_{33} 
\end{array}\right] \notag \\
&\label{eqn:R} \qquad \qquad \qquad \qquad \qquad \qquad \quad  \in \mcl{S}^{\beta(3\alpha+1)+2q(\alpha,\beta,d)}[x].
\end{align}

The following theorem states the conditions on the polynomial matrix $R(x)$ whose submatrices define  $R_b$ and $\bar{R}$ such that~\eqref{eqn:Rb_Rbar} is non-negative on $\mcl{H}^\alpha\left([0,1];\R^\beta \right)$.

\begin{theorem}\label{thm:non_neg}
Given (polynomial) matrices which define $\mcz{R}(w)$ in~\eqref{eqn:Rb_Rbar},
if
\begin{equation}\label{eqn:R_condition}
R(x) \succeq 0, \quad \forall x \in [0,1],
\end{equation} where $R(x)$ is defined in~\eqref{eqn:R}, then
\begin{equation}\label{eqn:R_conclusion}
\mcz{R}(w) \geq 0, \quad \forall w \in \mcl{H}^\alpha\left([0,1];\R^\beta \right).
\end{equation}
\end{theorem} 
%%%%%%%%%%%%%%%%%%%%%%%%%
%%%%%%%%%%%%%%%%%%%%%%%%%
\begin{proof}
Throughout this proof, for notational brevity, we write $q$ in place of $q(\alpha,\beta,d)$. 

If we define
\[f(x)=\bmat{
\wb(x) \\ 
\int_0^x Y_{q}(x,y)\w(y)dy \\
\int_x^1 Y_q(x,y)\w(y)dy
}, \quad w \in \mcl{H}^\alpha\left([0,1];\R^\beta\right),\]  then
\begin{align}\label{eqn:non_neg_rep}
&\int_0^1
f(x)^T R(x) f(x)dx \geq 0,\quad \forall w \in \mcl{H}^\alpha\left([0,1];\R^\beta\right),
\end{align} since $R(x) \succeq 0$, for all $x \in [0,1]$. Simplifying the expression we obtain
\begin{align}
&\int_0^1
f(x)^T R(x) f(x)dx \notag \\
&\label{eqn:non_neg:1} = \int_0^1 \wb(x)^T R_b(x) \wb(x)dx + \sum_{i=1}^3 \Theta_i,
\end{align} where
\begin{align*}
\Theta_1=& 2 \int_0^1 \int_0^x \w(x)^T R_{12}(x)Y_q(x,y)\w(y)dy dx \\
& + 2 \int_0^1 \int_x^1 \w(x)^T R_{13}(x)Y_q(x,y)\w(y)dy dx,\\
\Theta_2=& \int_0^1 \left(\int_0^x Y_q(x,y) \w(y)dy  \right)^T \\
& \times \left(\int_0^x R_{22}Y_q(x,y)\w(y)dy \right. \\
& \left. \qquad \qquad \qquad  + \int_x^1 R_{23}Y_q(x,y)\w(y)dy   \right)dx,\\
\Theta_3=& \int_0^1 \left(\int_x^1 Y_q(x,y) \w(y)dy  \right)^T \\
& \times \left(\int_0^x R_{23}^T Y_q(x,y)\w(y)dy \right. \\
& \left. \qquad \qquad \qquad  + \int_x^1 R_{33}Y_q(x,y)\w(y)dy   \right)dx.
\end{align*} Applying Lemma~\ref{lem:switch} to $\Theta_1$ with $K_1(x,y)=2R_{12}(x)Y_q(x,y)$ and $K_2(x,y)=2R_{13}(x)Y_q(x,y)$ produces
\begin{align}
\Theta_1 =&  \int_\Delta \w(x)^T \Gamma \left[R_{12}(x)Y_q(x,y) \right] \w(y) d\Delta \notag \\
&\label{eqn:non_neg:theta1}+\int_\Delta \w(x)^T \Gamma \left[ Y_q(y,x)^T R_{13}(y)^T \right] \w(y)d\Delta,
\end{align} Applying Lemma~\ref{lem:order} to $\Theta_2$ with
\begin{align*}
F_1(x,y)=&Y_q(x,y), \quad F_2(x,y) = 0_{q,\beta(\alpha+1)}, \quad v(y)=\w(y),\\
G_1(x,y)=& R_{22}Y_q(x,y), \quad G_2(x,y)=R_{23}Y_q(x,y),
\end{align*} produces
\begin{align}
\Theta_2 = & \int_\Delta \w(x)^T \Gamma \left[ \int_y^x Y_q(z,x)^T \frac{R_{23}^T}{2}Y_q(z,y)dz\right] \w(y) d\Delta, \notag \\
&\label{eqn:non_neg:theta2} +  \int_\Delta \w(x)^T \Gamma \left[ \int_x^1 Y_q(z,x)^T R_{22}Y_q(z,y)dz\right] \w(y) d\Delta.
\end{align} Similarly, applying Lemma~\ref{lem:order} to $\Theta_3$ with
\begin{align*}
F_1(x,y)=&0_{q,\beta(\alpha+1)}, \quad F_2(x,y) = Y_q(x,y), \quad v(y)=\w(y),\\
G_1(x,y)=& R_{23}^T Y_q(x,y), \quad G_2(x,y)=R_{33}Y_q(x,y),
\end{align*} produces
\begin{align}
\Theta_3 = & \int_\Delta \w(x)^T \Gamma \left[ \int_0^y Y_q(z,x)^T R_{33}^TY_q(z,y)dz\right] \w(y) d\Delta, \notag \\
&\label{eqn:non_neg:theta3} +  \int_\Delta \w(x)^T \Gamma \left[ \int_y^x Y_q(z,x)^T \frac{R_{23}^T}{2}Y_q(z,y)dz\right] \w(y) d\Delta.
\end{align} Substituting~\eqref{eqn:non_neg:theta1}-\eqref{eqn:non_neg:theta3} into~\eqref{eqn:non_neg:1} produces
\begin{align*}
\int_0^1 f(x)^T R(x)f(x)dx = & \int_0^1 \wb(x)^T R_b(x) \wb(x)dx \\
&+ \int_\Delta \w(x)^T \Gamma \left[ \bar{R}\right] \w(y)d\Delta
\end{align*} Then,~\eqref{eqn:non_neg_rep} completes the proof.
\end{proof} 
%%%%%%%%%%%%%%%%%%%%%%%%%%%%%%%%%%%%%%%%%%%%%%%%%%%%%%

We now present a corollary in which we construct strictly positive integral inequalities on $\mcl{L}_2\left([0,1];\R^\beta \right)$ that have the same form as $\mcz{V}(w)$ in~\eqref{eqn:V_int}. This corollary corresponds to setting $\alpha=0$ in Theorem~\ref{thm:non_neg}.  

Given $w \in \mcl{L}_2\left([0,1];\R^\beta \right)$ and $\beta,d \in \N$, let us define
\begin{align}
\mcz{T}(w) \hspace{-1mm} =  \hspace{-1mm} &\label{eqn:Tb_Tbar} \int_0^1 w(x)^T T_b(x) w(x)dx  \hspace{-1mm} + \hspace{-1mm} \int_\Delta w(x)^T \Gamma \left[\bar{T} \right] w(y)d\Delta,
\end{align} where $T_{b} \in \mcl{S}^{\beta}[x]$ and
\begin{align}
\bar{T}(x,y)=& T_{12}(x) Y_{q(0,\beta,d)}(x,y) + Y_{q(0,\beta,d)}(y,x)^T T_{13}(y)^T \notag  \\
& + \int_0^y Y_{q(0,\beta,d)}(z,x)^T T_{33} Y_{q(0,\beta,d)}(z,y)dz \notag \\
& + \int_y^x Y_{q(0,\beta,d)}(z,x)^T T_{23}^T Y_{q(0,\beta,d)}(z,y)dz \notag \\
&\label{eqn:Tbar}  +  \int_x^1 Y_{q(0,\beta,d)}(z,x)^T T_{22} Y_{q(0,\beta,d)}(z,y)dz.
\end{align} for any $T_{12},T_{13} \in \mcl{R}^{\beta \times q(0,\beta,d)}[x]$, $T_{22},T_{33} \in \S^{q(0,\beta,d)}$ and $T_{23} \in \R^{q(0,\beta,d) \times q(0,\beta,d)}$, which also define
\begin{equation}\label{eqn:T}
T(x)=\bmat{T_{b}(x) &  T_{12}(x) & T_{13}(x) \\ 
  T_{12}(x)^T  & T_{22} & T_{23} \\ 
 T_{13}(x)^T  & T_{23}^T & T_{33}} \in \mcl{S}^{\beta+2q(0,\beta,d)}[x].
\end{equation}

\begin{corollary}\label{cor:pos}
Given (polynomial) matrices which define $\mcz{T}(w)$ in~\eqref{eqn:Tb_Tbar},
 if there exists a positive scalar $\epsilon$ such that
\begin{equation}\label{eqn:T_condition}
T(x) -\bmat{\epsilon I_\beta & 0_{\beta,2q(0,\beta,d)} \\ 0_{2q(0,\beta,d),\beta} & 0_{2q(0,\beta,d)}} \succeq 0, \quad \forall x \in [0,1],
\end{equation} where $T(x)$ is defined in~\eqref{eqn:T}, then,
\begin{equation}\label{eqn:T_conclusion}
\mcz{T}(w) \geq \epsilon \norm{w}^2_{\mcl{L}_2}, \quad \forall w \in \mcl{L}_2\left([0,1];\R^\beta \right).
\end{equation}
\end{corollary}
%%%%%%%%%%%%%%%%%%%%
%%%%%%%%%%%%%%%%%%%%
\begin{proof}
Following the same steps as for the proof of Theorem~\ref{thm:non_neg}, it can be shown that
\begin{align*}
&\mcz{T}(w)-\epsilon \norm{w}_{\mcl{L}_2} \\
&= \int_0^1 f(x)^T \left(T(x) -\bmat{\epsilon I_\beta & 0_{\beta,2q(0,\beta,d)} \\ 0_{2q(0,\beta,d),\beta} & 0_{2q(0,\beta,d)}}   \right) f(x)dx,
\end{align*} where
\[f(x) = \bmat{w(x) \\ \int_0^x Y_{q(0,\beta,d)}(x,y) w(y)dy \\ \int_x^1 Y_{q(0,\beta,d)}(x,y) w(y)dy}, \quad w \in \mcl{L}_2\left([0,1];\R^\beta\right).\] Then,~\eqref{eqn:T_conclusion} holds since~\eqref{eqn:T_condition} holds.
\end{proof} 

%%%%%%%%%%%%%%%%%%%%%%%%%%%%%%%%%%%%%%%%%%%%%%%%%%%%%%%%%%%%%%%%%%%%%%%%%%%%%%%%%%%%%%%%
\section{SLACK INTEGRALS}\label{sec:slack}

In Section~\ref{sec:pro_statement} we cast the stability of~\eqref{eqn:PDE_complete} as a test of positivity and non-negativity of integral inequalities in~\eqref{eqn:stab:V_condition} and~\eqref{eqn:stab:Vd_condition}. In Section~\ref{sec:main} we will show that Corollary~\ref{cor:pos} is sufficient to verify~\eqref{eqn:stab:V_condition}. Theorem~\ref{thm:non_neg} is too conservative to verify~\eqref{eqn:stab:Vd_condition} because it enforces $\mcz{V}_d(w) \geq 0$ on the entire space $\mcl{H}^\alpha\left([0,1];\R^\beta \right)$, while we are only interested in non-negativity over the subset $\mcl{B} \subset \mcl{H}^\alpha\left([0,1];\R^\beta \right)$.
 As discussed in Section~\ref{sec:pro_statement}, however, to solve this problem we  construct slack integrals $\mcz{S}(w)$ which we defined as integral expressions with the same structure as $\mcz{V}_d(w)$ and satisfies
\begin{equation}\label{eqn:def:space}
\mcz{S}(w)=0, \quad \forall w \in \mcl{B}.
\end{equation} Then, we may use Theorem~\ref{thm:non_neg} to test if
\[\mcz{V}_d(w)+\mcz{S}(w) \geq 0, \quad \forall w \in \mcl{H}^\alpha\left([0,1];\R^\beta \right),\] which, if true, would imply that
$\mcz{V}_d(w) \geq 0$, $\forall w \in \mcl{B}$, owing to~\eqref{eqn:def:space}.

We will construct slack integrals using quadratic forms of the Fundamental Theorem of Calculus (FTC) and Green's theorem.
%%%%%%%%%%%%%%%%%%%%%%%%
\begin{lemma}[FTC quadratic form]\label{lem:FTC}
For any $K_1 \in  \mcl{R}^{\beta\alpha \times \beta\alpha}[x]$ and $K_2 \in \mcl{R}^{\beta \alpha \times 2 \beta \alpha}[x]$ the following identity holds
\begin{equation}\label{eqn:FTC_quad}
\int_0^1  \wb(x)^T He \left( K_b(x) \right)  \wb(x) dx=0, \quad \forall w \in \mcl{H}^\alpha\left([0,1];\R^\beta \right),
\end{equation} where
\begin{align*}
K_b= & \bmat{K_{b1}(x) & K_{b2}(x) \\ 
 0_{2\beta\alpha,\beta(\alpha+1)} & K_{b3} 
}, \\
K_{b1}(x)=& N^T \pfx K_1(x) N +N^T K_1(x) N_\partial + N_\partial^T K_1(x) N,\\
K_{b2}(x)=& N_\partial^T K_2(x) + N^T \pfx K_2(x),\\
K_{b3}=& N_0^T K_1(0) N_0-N_1^T K_1(1) N_1 + N_0 K_2(0) - N_1 K_2(1).
\end{align*}
\end{lemma} 
%%%%%%%%%%%%%%
\begin{proof}
The identity~\eqref{eqn:FTC_quad} is established by expanding
\[\int_0^1 \frac{d}{dx}g(x)dx - \left(g(1)-g(0) \right)=0,\] with
\begin{align*}
g(x)=& w_{\alpha-1}(x)^T K_1(x) w_{\alpha-1}(x) + w_{\alpha-1}(x)^T K_2(x)  w_\alpha^b.
\end{align*}
\end{proof}

Next, we present the quadratic form of the Green's theorem. The proof of the following lemma is provided in the Appendix.
%%%%%%%%%%%%%%%%%%%%%
\begin{lemma}[Green's theorem quadratic form]\label{lem:Greens}
For any $H_1,H_2 \in \mcl{R}^{\beta\alpha \times \beta\alpha}[(x,y)]$, the following identity holds
\begin{align*}
&\int_0^1 \wb(x)^T He \left( H_b(x) \right) \wb(x) dx \\
&  +\int_\Delta \w(x)^T \Gamma \left[ \bar{H} \right] \w(y)d\Delta  =0, \quad \forall w \in \mcl{H}^\alpha\left([0,1];\R^\beta \right),
\end{align*} where
\begin{align*}
H_b(x)=&
\bmat{
H_{b1}(x)  
&
H_{b2}(x)  \\ 
H_{b3}(x) 
&
0_{2 \beta \alpha}},\\
H_{b1}(x)=&-N^T \left(H_1(x,x)+H_2(x,x) \right) N,\\
H_{b2}(x)=&N^T H_1(x,0)N_0,\quad H_{b3}(x)=N_1^T H_2(1,x) N,\\
\bar{H}(x,y)=& \hlf  N^T \left(\pfy H_1(x,y) -\pfx H_2(x,y)  \right) N \\
&+\hlf \left( N^T H_1(x,y) N_\partial - N_\partial^T H_2(x,y) N  \right).
\end{align*}
\end{lemma}

%%%%%%%%%%%%%%%%%%%%%%%%%%%%%
In the following lemma we formulate an integral equation that holds on the set $\mcl{B}$ given in~\eqref{eqn:BC}. The proof of the following lemma is provided in the Appendix.
\begin{lemma}\label{lem:boundary}
Given $F \in \R^{\beta\alpha \times 2 \beta\alpha}$ the following identity holds true for all $w \in \mcl{B}$:
\begin{align*}
&\int_0^1 \wb(x)^T He \left( B_b(x) \right) \wb(x)dx = 0, 
\end{align*} where
\begin{align*}
B_b(x) = & \bmat{0_{\beta(\alpha+1)} & B_1(x)F \\ 0_{2\beta\alpha,\beta(\alpha+1)} & B_2(x) F},
\end{align*} for any  $B_1 \in \mcl{R}^{\beta(\alpha+1) \times \beta\alpha}[x]$ and $B_2 \in \mcl{R}^{2\beta\alpha \times \beta\alpha}[x]$.
\end{lemma}
%%%%%%%%%

We now use the results in Lemmas~\ref{lem:FTC}-\ref{lem:boundary} to formulate slack integrals on the set $\mcl{B} \in \mcl{H}^\alpha\left([0,1];\R^\beta \right)$.

Let us define
\begin{align}
\mcz{S}(w)=&\int_0^1 \wb(x)^T He\left(K_b(x)+H_b(x)+B_b(x)  \right)\wb(x)dx \notag \\
&\label{eqn:slack_integral}  + \int_\Delta \w(x)^T \Gamma \left[\bar{H}  \right] \w(y) d\Delta, 
\end{align} where $K_b$ is parameterized by $K_1$ and $K_2$ as in Lemma~\ref{lem:FTC}, $H_b$ and $\bar{H}$ are parameterized by $H_1$ and $H_2$ as in Lemma~\ref{lem:Greens} and $B_b$ is parameterized by $B_1$ and $B_2$, and the matrix $F$ which defines the set $\mcl{B}$ as in Lemma~\ref{lem:boundary}.

We now state the main result of this section.
\begin{theorem}\label{thm:slack_equality}
Given matrix $F$ which defines the set $\mcl{B}$ in~\eqref{eqn:BC}, the following identity holds true
\begin{equation}\label{eqn:slack_equality}
\mcz{S}(w)=0, \quad \forall w \in \mcl{B},
\end{equation} where $\mcz{S}(w)$ is parameterized by any $K_i$, $H_i$ and $B_i$, $i \in \{1,2\}$, as in  in~\eqref{eqn:slack_integral}.
\end{theorem}
%%%%%%%%%%%%%%%%%%%%%%%%%%%%%%%%%
\begin{proof}
We begin by considering the following decomposition
\begin{align}
\mcz{S}(w)=&\int_0^1 \wb(x)^T He\left(K_b(x)+H_b(x)+B_b(x)  \right)\wb(x)dx \notag \\
&\label{eqn:slack:decomp}  + \int_\Delta \w(x)^T \Gamma \left[\bar{H}  \right] \w(y) d\Delta   = \sum_{i=1}^3 \Theta_i,
\end{align} where
\begin{align*}
\Theta_1 = & \int_0^1 \wb(x)^T He\left( K_b(x) \right)\wb(x)dx, \\
 \Theta_2 = & \int_0^1 \wb(x)^T He \left( H_b(x) \right) \wb(x)dx \\
 &  + \int_\Delta \w(x)^T \Gamma \left[\bar{H}  \right] \w(y) d\Delta, \\
\Theta_3 = & \int_0^1 \wb(x)^T He \left( B_b(x) \right) \wb(x).
\end{align*} From Lemmas~\ref{lem:FTC} and~\ref{lem:Greens} we have that
\begin{equation}\label{eqn:slack:decomp1}
\Theta_1 = 0 \quad \text{and} \quad \Theta_2=0, \quad \forall w \in \mcl{H}^\alpha\left([0,1];\R^\beta \right).
\end{equation} From Lemma~\ref{lem:boundary} we have that
\begin{equation}\label{eqn:slack:decomp2}
\Theta_3 = 0, \quad \forall w \in \mcl{B}.
\end{equation} Therefore, from~\eqref{eqn:slack:decomp}-\eqref{eqn:slack:decomp2} we conclude that the expression in~\eqref{eqn:slack_equality} holds for all $w \in \mcl{B}$.
\end{proof}
 
%%%%%%%%%%%%%%%%%%%%%%%%%%%%%%%%%%%%%%%%%%%%%%%%%%%%%%%%%%%%%%%%%%%%%%%%%%%%%%%%%%%%%%%%
\section{MAIN RESULT}\label{sec:main}

In this section we use the results formulated in Sections~\ref{sec:non_neg_ineq}-\ref{sec:slack} to construct a method of verifying the stability of PDE~\eqref{eqn:PDE_complete}. 
Let us proceed with the following.
%%%%%%%%%%%%%%%%%%%%%%%%%%%%%%%%%%%%%%%%%%%
\begin{theorem}\label{thm:main}
Consider the PDE~\eqref{eqn:PDE_complete}. For any given $d \in \N$, positive scalars $\epsilon,\delta$ and polynomial $Y_{q(\alpha,\beta,d)} \in \mcl{R}^{q(\alpha,\beta,d) \times \beta(\alpha+1)}[(x,y)]$ defined in~\eqref{eqn:Y}, suppose there exist: 
\begin{subequations}\label{eqn:vars}
\begin{align}
&\label{eqn:var1}\text{(polynomial) matrices defining $\mcz{R}(w)$ in~\eqref{eqn:Rb_Rbar}},\\
&\label{eqn:var2}\text{(polynomial) matrices defining $\mcz{T}(w)$ in~\eqref{eqn:Tb_Tbar}},\\
&\label{eqn:var3}S_T \in \mcl{S}^{\beta+2q(0,\beta,d)}[x], \quad S_R \in \mcl{S}^{\beta(3\alpha+1)+2q(\alpha,\beta,d)}[x], \\
&\label{eqn:var4}K_1 \in \mcl{R}^{\beta \alpha \times \beta \alpha}[x], \quad K_2 \in \mcl{R}^{\beta \alpha \times 2\beta \alpha}[x],\\ 
&\label{eqn:var5}H_1,H_2 \in \mcl{R}^{\beta \alpha \times \beta \alpha}[(x,y)], \quad S \in \mcl{S}^{\beta(3\alpha+1)}[x],\\
&\label{eqn:var6}B_1 \in \mcl{R}^{\beta(\alpha+1) \times \beta \alpha}[x], \quad  B_2 \in \mcl{R}^{2\beta \alpha \times \beta \alpha}[x],
\end{align}
\end{subequations} such that
\begin{subequations}\label{eqn:constraints}
\begin{align}
& T(x) - \bmat{\epsilon I_\beta & 0_{\beta,2q(0,\beta,d)} \\ 0_{2q(0,\beta,d),\beta} & 0_{2q(0,\beta,d)}} \notag \\
&\label{eqn:constraint2}\qquad \qquad \qquad \qquad-S_T(x) \omega(x) \in \Sigma^{\beta + 2q(0,\beta,d)}[x],\\
&\label{eqn:constraint2a} S_T \in \Sigma^{\beta+2q(0,\beta,d)}[x],\\
&\label{eqn:constraint4} R(x) - S_R(x) \omega(x) \in \Sigma^{\beta(3\alpha+1)+2 q(\alpha,\beta,d)}[x],\\
&\label{eqn:constraint4a} S_R(x)  \in \Sigma^{\beta(3\alpha+1)+2 q(\alpha,\beta,d)}[x],\\
&He \left(U_b(x)+K_b(x)+H_b(x)+B_b(x)  \right) \notag \\
&\label{eqn:constraint5} \qquad \qquad \qquad \qquad  -R_b(x)-S(x) \omega (x) \in \Sigma^{\beta(3\alpha+1)}[x],\\
&\label{eqn:constraint5a} S(x)  \in \Sigma^{\beta(3\alpha+1)}[x],\\
&\label{eqn:constraint6} \bar{U}(x,y)+\bar{H}(x,y)-\bar{R}(x,y) = 0_{\alpha+1}, 
\end{align}
\end{subequations} where $\omega(x)=x(1-x)$, $T(x)$ and $R(x)$ are defined in~\eqref{eqn:T} and~\eqref{eqn:R}, respectively, $K_b$ is defined using $K_i$ as in Lemma~\ref{lem:FTC}, $H_b$ and $\bar{H}$ are defined using $H_i$ as in Lemma~\ref{lem:Greens}, $B_b$ is defined using $B_i$ and $F$ as in Lemma~\ref{lem:boundary}, and polynomials $U_b$ and $\bar{U}$ are defined using $\delta$ in~\eqref{eqn:Vd_int} and with $T_b$ and $\bar{T}$ defined in~\eqref{eqn:Tb_Tbar}.

Then,~\eqref{eqn:PDE_complete} is exponentially stable.
\end{theorem}
%%%%%%%%%%%%%%%%%%%%%%%%%%%%%%%%
\begin{proof}
Since the polynomial matrix in~\eqref{eqn:constraint2} is Sum-of-Squares (SOS), $S_T$ is SOS in~\eqref{eqn:constraint2a} and $\omega(x) \geq 0$, for all $x \in [0,1]$, using the property that a SOS polynomial is non-negative for all $x \in \R$, we conclude that $T(x)$ satisfies~\eqref{eqn:T_condition}. Therefore,
$\mcz{T}(w) \geq \epsilon \norm{w}_{\mcl{L}_2}$, for all $w \in \mcl{L}_2\left([0,1];\R^\beta \right)$. Moreover, from~\eqref{eqn:V_int} we have that $\mcz{V}(w) = \hlf \mcz{T}(w)$ and thus
\begin{equation}\label{eqn:main:conclusion1}
\mcz{V}(w) \geq \mu \norm{w}_{\mcl{L}_2}, \quad \forall w \in \mcl{L}_2\left([0,1];\R^\beta \right),
\end{equation} with $\mu = \hlf \epsilon$. Similarly, since~\eqref{eqn:constraint4}-\eqref{eqn:constraint4a} hold, we conclude from Theorem~\ref{thm:non_neg} that
\begin{equation}\label{eqn:main:R}
\mcz{R}(w) \geq 0, \quad \forall w \in \mcl{H}^\alpha\left([0,1];\R^\beta \right).
\end{equation} 

Now, let us define
\begin{align}
\mcz{P}(w) \hspace{-1mm}= \hspace{-1mm} &\label{eqn:main:P_int} \int_0^1 \wb(x)^T P_b(x) \wb(x)dx  + \int_\Delta \w(x)^T \Gamma \left[\bar{P} \right] \w(y)d\Delta,
\end{align} where
\begin{align*}
P_b(x)=&He \left(U_b(x)+K_b(x)+H_b(x)+B_b(x) \right)   - R_b(x) ,\\
\bar{P}(x,y)=& \bar{U}(x,y)+\bar{H}(x,y)-\bar{R}(x,y). 
\end{align*}
Since~\eqref{eqn:constraint5}-\eqref{eqn:constraint6} hold, we deduce that
\begin{align*}
P_b(x) &\succeq 0, \quad \forall x \in [0,1],\quad \bar{P}(x,y)&=0, \quad \forall (x,y) \in \Delta.
\end{align*}  Therefore, we get
\begin{equation}\label{eqn:main:P_condition}
\mcz{P}(w) \geq 0, \quad \forall w \in \mcl{H}^\alpha\left([0,1];\R^\beta \right).
\end{equation} From the definition of $\mcz{P}(w)$ in~\eqref{eqn:main:P_int} it is clear that
$\mcz{P}(w)=\mcz{V}_d(w) + \mcz{S}(w) - \mcz{R}(w)$, where $\mcz{V}_d(w)$ is defined in~\eqref{eqn:Vd_int} and $\mcz{S}(w)$ is defined in~\eqref{eqn:slack_integral}. Thus, using~\eqref{eqn:main:P_condition} we obtain
\[\mcz{V}_d(w) + \mcz{S}(w) - \mcz{R}(w) \geq 0, \quad \forall w \in \mcl{H}^\alpha\left([0,1];\R^\beta \right).\]

Using~\eqref{eqn:main:R} the previous expression may be reduced to
\begin{equation}\label{eqn:main:cond1}
\mcz{V}_d(w) + \mcz{S}(w)  \geq 0, \quad \forall w \in \mcl{H}^\alpha\left([0,1];\R^\beta \right).
\end{equation} Now, from Theorem~\ref{thm:slack_equality} we have that $\mcz{S}(w)=0$, for all $w \in \mcl{B}$, thus, from~\eqref{eqn:main:cond1} we deduce that
\begin{equation}\label{eqn:main:conclusion2}
\mcz{V}_d(w) \geq 0, \quad \forall w \in \mcl{B} \subset \mcl{H}^\alpha\left([0,1];\R^\beta \right).
\end{equation} Since~\eqref{eqn:main:conclusion1} and~\eqref{eqn:main:conclusion2} hold, we apply Theorem~\ref{thm:analysis} to complete the proof. 
\end{proof}

Since polynomials are closed under differentiation and integration, we have that $K_b$ and $\bar{H}$ containing $\pfx K_1(x)$, $\pfx K_2(x)$, $\pfy H_1(x,y)$ and $\pfx H_2(x,y)$ in Lemmas~\ref{lem:FTC}-\ref{lem:Greens} are polynomials. Moreover, since the polynomial $Y_{q(\alpha,\beta,d)}(x,y)$ is fixed, the terms $R_b$, $\bar{R}$, $T_b$ and $\bar{T}$ in Theorem~\ref{thm:non_neg} and Corollary~\ref{cor:pos} are polynomials affine in their respective (polynomial) matrices. Therefore, the conditions in~\eqref{eqn:constraints} are simply either, $1$) a verification of the membership of polynomial matrices in the set of SOS polynomials as in~\eqref{eqn:constraint2}-\eqref{eqn:constraint5a}, or, $2$) enforcement of affine constraints on the polynomial variables as in~\eqref{eqn:constraint6}. Indeed, the problem of searching for SOS polynomials subject to affine constraints is a Semi-Definite Program (SDP),~\cite[Chapter~3]{blekherman2012semidefinite},~\cite{vandenberghe1996semidefinite}. We are then interested in solving
\begin{equation}\label{eqn:poly:convex_opt}
\textbf{SDP Problem:} \quad \text{Find}~\eqref{eqn:vars}~\text{subject to}~\eqref{eqn:constraints}.
\end{equation}
The numerical implementation is performed by constructing the underlying SDP for~\eqref{eqn:poly:convex_opt} by using the freely available packages SOSTOOLS~\cite{sostools} or YALMIP~\cite{lofberg2005yalmip}. Then, the associated SDP is solved, for example, using SeDuMi~\cite{sturm1999using} or SDPA~\cite{yamashita2003} solvers.

%%%%%%%%%%%%%%%%%%%%%%%%%%%%%%%%%%%%%%%%%%%%%%%%%%%%%%%%%%%%%%%%%%%%%%%%%%%%%%%%%%%%%%%%
\subsection{Numerical Examples}\label{sec:num}

We now determine the stability of PDEs~\eqref{eqn:exmp:parabolic},~\eqref{eqn:exmp:hyperbolic} and~\eqref{eqn:exmp:beam} presented in Section~\ref{sec:pro_statement} by solving~\eqref{eqn:poly:convex_opt} for each of these systems. Recall that each of these systems are parameterized by a positive scalar $\lambda \in \R$. The parameter $\lambda$ can alter the stability properties of these systems, and thus, allows us to measure the effectiveness of the proposed methodology.

We run the numerical codes for the example PDEs with $\epsilon = 10^{-3}$, $\delta= 10^{-4}$ and maximum polynomial degree allowed by the memory resources of $8$ giga bytes of RAM. In the following examples, we search for the maximum $\lambda \in \R$ for which the conditions of problem~\eqref{eqn:poly:convex_opt} are feasible using a bisection search with a resolution of $10^{-3}$. These results are provided in Table~\ref{tab:1}. The discussion of the results is provided below.

\begin{table}
\centering
\begin{tabular}{ |c|c|c|c|c|} 
 \hline
 & $d=2$ & $4$ & $6$ & $8$ \\ \hline 
 Example~$1$ Eqn.~\eqref{eqn:exmp:parabolic} & $\lambda=3.263 $ & $3.263$ & $3.409$ & $3.409$ \\ 
 Example~$2$ Eqn.~\eqref{eqn:exmp:hyperbolic} & Inf. & Inf. & $\lambda = 0.999$ & $0.999$ \\
 Example~$3$ Eqn.~\eqref{eqn:exmp:beam} & $\lambda=0.999$ & $0.999$ & $0.999$ & - \\
 \hline
\end{tabular}
\caption{Maximum $\lambda \in \R$ for which problem~\eqref{eqn:poly:convex_opt} is feasible for Examples~$1$-$3$ in Equations~\eqref{eqn:exmp:parabolic}-\eqref{eqn:exmp:beam}, respectively, as a function of polynomial degree $d$. Here, Inf. denotes infeasibility.}
\label{tab:1}
\end{table}

\subsubsection*{Example~$1$} Using finite-differences with $1500$ uniformly distributed spatial points we approximate that the parabolic PDE in~\eqref{eqn:exmp:parabolic} is stable for $\lambda <  3.412$.  As illustrated in Table~\ref{tab:1}, the maximum $\lambda$ for which the problem~\eqref{eqn:poly:convex_opt} is feasible is $\lambda=3.409$ which is $99.91 \%$ of the value of $3.412$ obtained via the finite-difference approximation.

\subsubsection*{Example~$2$} Using~\cite[Lemma~3.1]{di2013stabilization} it can be established that the coupled first order hyperbolic PDE in~\eqref{eqn:exmp:hyperbolic} is exponentially stable for $\lambda < 1$.  As illustrated in Table~\ref{tab:1} the maximum $\lambda$ for which problem~\eqref{eqn:poly:convex_opt} is feasible is $\lambda=0.999$ which is $99.9 \%$ of the stability value of $1$.

\subsubsection*{Example~$3$} Finally, the Euler-Bernoulli beam model in~\eqref{eqn:exmp:beam} is stable for $\lambda<1$ (see~\cite[Exercise~8.3]{krstic2008boundary}). From Table~\ref{tab:1} we observe that the maximum $\lambda$ for which problem~\eqref{eqn:poly:convex_opt} is feasible is $\lambda=0.999$ which is $99.9 \%$ of the stability value of $1$. Note that  owing to memory constraints, we could only solve the problem ~\eqref{eqn:poly:convex_opt} for a maximum degree of $d=6$.  

Note that our method successfully establishes the stability of the example PDEs within $99.9 \%$ of the approximated/analytic stability margin for $\lambda$. Finally, the time taken to search for variables which solve the problem~\eqref{eqn:poly:convex_opt} is provided in Table~\ref{tab:2}. 

\begin{table}
\centering
\begin{tabular}{ |c|c|c|c|c|} 
 \hline
   &   $d=2$ & $d=4$ & $d=6$ & $d=8$ \\ \hline 
 Example~$1$ Eqn.~\eqref{eqn:exmp:parabolic} & $4.415$ & $8.111$ & $20.273$ & $48.080$  \\ 
 Example~$2$ Eqn.~\eqref{eqn:exmp:hyperbolic} & $7.085$ & $12.796$ & $49.751$ & $141.297$  \\
 Example~$3$ Eqn.~\eqref{eqn:exmp:beam} & $25.663$ & $113.708$ & $360.876$ & -  \\
 \hline
\end{tabular}
\caption{Computer run time (in seconds) for performing the search for variables which solve problem~\eqref{eqn:poly:convex_opt} for Examples~$1$-$3$ in Equations~\eqref{eqn:exmp:parabolic}-\eqref{eqn:exmp:beam}, respectively.}
\label{tab:2}
\end{table}
%%%%%%%%%%%%%%%%%%%%%%%%%%%%%%%%%%%%%%%%%%%%%%%%%%%%%%%%%%%%%%%%%%%%%%%%%%%%%%%%%%%%%%%%
\section{CONCLUSION}

We presented a method to verify stability for a large class of 1-D PDEs with polynomial data. The presented methodology relies on using Lyapunov's method to reduce the problem of stability to the verification of integral inequalities. An application of the fundamental theorem of calculus and Green's theorem allows us to formulate a polynomial problem for verifying the integral inequalities. Using the properties of SOS polynomials allowed us to solve the polynomial problem in a computationally efficient manner by casting the problem as an SDP. The numerical experiments indicate that the method can predict the stability of the systems considered up to a high degree of accuracy. We would like to extend this method to consider an even larger class of PDEs, for example, by including Partial (Integro)-Differential Equations (P(I)DEs) and boundary feedback. Furthermore, we would like to formulate this theory for general PDEs, i.e., PDEs not constrained to have polynomial data. Eventually, we would like to extend this framework to in-domain/boundary controller synthesis for PDEs.

%%%%%%%%%%%%%%%%%%%%%%%%%%%%%%%%%%%%%%%%%%%%%%%%%%%%%%%%%%%%%%%%%%%%%%%%%%%%

%%%%%%%%%%%%%%%%%%%%%%%%%%%%%%%%%%%%%%%%%%%%%%%%%%%%%%%%%%%%%%%%%%%%%%%%%%%%%%%%%%%%%%%%

\appendix
\setcounter{lemma}{0}
\renewcommand{\thelemma}{\Alph{section}.\arabic{lemma}}
%%%%%%%%%%%%%%%%%%%%%%%%%%%%%%%%%%%%%%%%%%%%%%%%%%%%%%%%%%%%%%%%%%%%%%%%%%%%%%%%%%%%%%%%%
In the Appendix we show how $\mcz{V}_d$ is cast as~\eqref{eqn:Vd_int} and we provide the proofs of the results used in this paper.

The following Lemma is used throughout the exposition.
\begin{lemma}\label{lem:switch}
For any bivariate matrices $K_1(x,y),K_2(x,y) \in \R^{\beta(\alpha+1) \times \beta(\alpha+1)}$, the following identity holds
\begin{align*}
&\int_0^1 \int_0^x \w(x)^T K_1(x,y) \w(y)dy dx \\
&  + \int_0^1 \int_x^1 \w(x)^T K_2(x,y) \w(y)dy dx \\
& = \hlf \int_\Delta \w(x)^T  \Gamma \left[ K_1(x,y) +  K_2(y,x)^T  \right] \w(y) d\Delta.  
\end{align*}
\end{lemma}
%%%%%%%%%%%%%
\begin{proof}
We have
\begin{align*}
&\int_0^1 \int_0^x \w(x)^T K_1(x,y) \w(y)dy dx \\
&=  \int_0^1 \int_0^x \w(y)^T K_1(x,y)^T \w(x)dy dx \\
&= \int_0^1 \int_y^1 \w(y)^T K_1(x,y)^T \w(x)dx dy \\
& = \int_0^1 \int_x^1 \w(x)^T K_1(y,x)^T \w(y)dy dx, 
\end{align*} where we first transposed the integrand, followed by a change of order of integration and finally switched between the variables $x$ and $y$. Thus
\begin{align*}
&\int_0^1 \int_0^x \w(x)^T K_1(x,y) \w(y)dy dx \\
& = \hlf  \int_0^1 \int_0^x \w(x)^T K_1(x,y) \w(y)dy dx \\
& \quad + \hlf \int_0^1 \int_x^1 \w(x)^T K_1(y,x)^T \w(y)dy dx \\
& = \hlf \int_\Delta \w(x)^T \Gamma \left[K_1(x,y)\right] \w(y)d\Delta.
\end{align*} Following the same steps for
\[\int_0^1 \int_x^1 \w(x)^T K_2(x,y) \w(y)dy dx,\] then completes the proof.
\end{proof}

%%%%%%%%%%%%%%%%%%%%%%%%%%%%%%%%%%%%%%%%%%%%%%%%%%%%%%%%

%%%%%%%%%%%%%%%%%%%%%%%%%%%%%%%%%%%%%%%%%%%%%%%%%%%%%%%%%%%%%%%%%%%%%%%%%%%%%%%%%%%%%%%%%%%

We now show how~\eqref{eqn:Vd_int} is formulated. We begin by writing the integral expression in~\eqref{eqn:V_int} as\footnote{For brevity we have dropped the temporal dependency of $w$.}
\begin{equation}\label{eqn:LF:pre_decomp1}
\mcz{V}(w)=\hlf \ip{\Xi w}{ w}_{\mcl{L}_2},
\end{equation} where the self-adjoint operator $\Xi$ on $\mcl{L}_2\left([0,1];\R^\beta\right)$ is defined as
\begin{align*}
\left(\Xi w \right)(x)=& T_b(x) w(x) + \int_0^x \bar{T}(x,y)w(y)dy \\
& + \int_x^1 \bar{T}(y,x)^T w(y)dy. 
\end{align*} Since the operator $\Xi$ is self-adjoint, we may use~\eqref{eqn:PDE_complete} to obtain
\begin{align*}
\pft \mcz{V}(w)
=& \hlf \ip{\Xi \pft w}{ w}_{\mcl{L}_2} + \hlf \ip{\Xi  w}{ \pft w}_{\mcl{L}_2} \\
=&\hlf \ip{\Xi  w}{ \pft w}_{\mcl{L}_2} + \hlf \ip{\Xi  w}{ \pft w}_{\mcl{L}_2} \\
=&\ip{\Xi  w}{ \pft w}_{\mcl{L}_2} \\
=&\int_0^1 \left(  T_b(x) w(x)dx + \int_0^x \bar{T}(x,y)w(y)dy \right. \notag \\
& \left.  \quad  + \int_x^1 \bar{T}(y,x)^T w(y) dy \right)^T   A(x)\w(x)dx .
\end{align*}

Therefore,
\begin{align}
\mcz{V}_d(w)=&-\pft \mcz{V}(w)-2\delta \mcz{V}(w) \notag \\
=&-\int_0^1 \left(  T_b(x) w(x)dx + \int_0^x \bar{T}(x,y)w(y)dy \right. \notag \\
& \left.  \quad  + \int_x^1 \bar{T}(y,x)^T w(y) dy \right)^T   A(x)\w(x)dx \notag \\
& \quad-\delta \int_0^1  w(x)^T T_b(x) w(x)dx \notag \\
&   \quad - \delta \int_\Delta w(x)^T \Gamma \left[\bar{T}(x,y) \right] w(y) d\Delta \notag \\
&\label{eqn:LF:decomp}=- \Phi_1 - \Phi_2,
\end{align} where,
\begin{align*}
\Phi_1=& \int_0^1 \left(  T_b(x) w(x) + \int_0^x \bar{T}(x,y)w(y)dy \right. \notag \\
& \left.    + \int_x^1 \bar{T}(y,x)^T w(y) dy \right)^T   A(x)\w(x)dx,\\
\Phi_2=& \delta \int_0^1  w(x)^T T_b(x) w(x)dx  \\
&   + \delta \int_\Delta w(x)^T \Gamma \left[\bar{T}(x,y) \right] w(y) d\Delta.
\end{align*} The term $\Phi_1$ may be written as
\begin{align*}
\Phi_1=& \int_0^1 \wb(x)^T \bmat{T_b(x)A(x) & 0_{\beta,2\beta \alpha} \\ 0_{3\beta\alpha,\beta(\alpha+1)} & 0_{3\beta\alpha,2\beta\alpha}}\wb(x)dx \\
&+ \int_0^1 \int_0^x \w(x)^T \bmat{\bar{T}(x,y)^T A(x) \\ 0_{\beta\alpha,\beta(\alpha+1)}}^T \w(y)dy dx \\
&+ \int_0^1 \int_x^1 \w(x)^T \bmat{\bar{T}(y,x)^T A(x) \\ 0_{\beta\alpha,\beta(\alpha+1)}}^T \w(y)dy dx,
\end{align*} where we have transposed the two double integrals. Then, applying Lemma~\ref{lem:switch} to the double integrals and writing the single integral kernel in a symmetric form produces
\begin{align}
&\Phi_1 = \notag \\
& \int_0^1 \wb(x)^T He \left( \bmat{T_b(x)A(x) & 0_{\beta,2\beta \alpha} \\ 0_{3\beta\alpha,\beta(\alpha+1)} & 0_{3\beta\alpha,2\beta\alpha}} \right) \wb(x)dx \notag \\
&+ \int_\Delta  \w(x)^T \Gamma\left[ \hlf  \bmat{\bar{T}(x,y)^T A(x) \\ 0_{\beta\alpha,\beta(\alpha+1)}}^T \right.  \notag \\
&\label{eqn:Phi1} \qquad \qquad \qquad \quad  \left. + \hlf \bmat{\bar{T}(x,y) A(y) \\ 0_{\beta\alpha,\beta(\alpha+1)}} \vphantom{\bmat{\bar{T}(x,y) A(y) \\ 0_{\beta\alpha,\beta(\alpha+1)}}^T} \right] \w(y) d\Delta .
\end{align}

The term $\Phi_2$ may be written as
\begin{align}
\Phi_2= & \int_0^1 \wb(x)^T He \left(\bmat{\delta T_b(x) & 0_{\beta,3\beta\alpha} \\ 0_{3\beta\alpha,\beta} & 0_{3\beta\alpha}}    \right) \wb(x)dx \notag \\
&\label{eqn:Phi2} +\hlf \int_\Delta \w(x)^T \Gamma \left[\bmat{2\delta\bar{T}(x,y) & 0_{\beta,\beta\alpha} \\ 0_{\beta\alpha,\beta} & 0_{\beta\alpha}} \right] \w(y)d\Delta.
\end{align} Substituting~\eqref{eqn:Phi1} and~\eqref{eqn:Phi2} into~\eqref{eqn:LF:decomp} gives~\eqref{eqn:Vd_int}.

%%%%%%%%%%%%%%%%%%%%%%%%%%%%%%%%%%%%%%%%%%%%%%%%%%%%%%%%%%%%%%%%%%%%%%%%%%%%%%%%%%%%%%%%%

We now provide the proofs of Lemmas~\ref{lem:Greens} and~\ref{lem:boundary}.

\begin{proof}[Proof of Lemma~\ref{lem:Greens}]
Consider the vector field
\[
\bmat{\phi_1(x,y) \\ \phi_2(x,y)} = \bmat{ w_{\alpha-1}(x)^T H_1(x,y) w_{\alpha-1}(y) \\ w_{\alpha-1}(x)^T H_2(x,y) w_{\alpha-1}(y)},
\]
 Then, by Green's theorem
 \begin{align*}
 &\oint_{\partial \overline{\Omega}} \left(\phi_1(x,y)dx + \phi_2(x,y)dy  \right) \\
 &  + \int_{\overline{\Omega}} \left(\pfy \phi_1(x,y) - \pfx \phi_2(x,y) \right)d \overline{\Omega} = 0,
 \end{align*} where $\partial\overline{\Omega}$ denotes the boundary of the domain $\overline{\Omega}$. Therefore, we obtain
\begin{align}
 & \int_0^1 \left[\phi_1(x,0)+\phi_2(1,x)-\phi_1(x,x)-\phi_2(x,x)  \right] dx \notag \\
 &\label{eqn:calc:initial}   +
\int_0^1 \int_0^x \left[ \pfy \phi_1-\pfx \phi_2  \right]dy dx=0.
\end{align}

Using the definitions of the projection matrices described in Section~\ref{sec:notation}, we obtain 
\begin{align}
&\int_0^1 \left[\phi_1(x,0)+\phi_2(1,x)-\phi_1(x,x)-\phi_2(x,x)  \right] dx \notag \\
&  =\int_0^1
\wb(x)^T
H_b(x) \wb(x)dx \notag \\
&\label{eqn:calc:boundary}  =\int_0^1
\wb(x)^T
He \left(H_b(x) \right) \wb(x)dx.
\end{align} Similarly, we also obtain the following identity
\begin{align*}
&\int_0^1 \int_0^x \left[ \pfy \phi_1(x,y)-\pfx \phi_2(x,y)  \right]dy dx \\
&  = \int_0^1 \int_0^x \w(x)^T 2 \bar{H}(x,y) \w(y)dy dx.
\end{align*} Then, applying Lemma~\ref{lem:switch} with $K_1 = 2\bar{H}$ and $K_2 = 0$ produces
\begin{align}
&\int_0^1 \int_0^x \w(x)^T 2 \bar{H}(x,y) \w(y)dy dx \notag \\
&\label{eqn:calc:indomain_1} \qquad = \int_\Delta  \w(x)^T \Gamma \left[ \bar{H}(x,y) \right] \w(y)dy dx.
\end{align} Substituting~\eqref{eqn:calc:boundary} and~\eqref{eqn:calc:indomain_1} into~\eqref{eqn:calc:initial} completes the proof.
\end{proof}

%%%%%%%%%%%%%%%%%%%%%%%%%%%%%%%%%%%%%%%%%%%%%%%%%%%%%%%%%%%%%%%%%%5
\begin{proof}[Proof of Lemma~\ref{lem:boundary}]
Since for all $w \in \mcl{B}$,
\[ F w_\alpha^b=0_{\beta\alpha,1},\] we get for all $w \in \mcl{B}$ 
\begin{align*}
0=&\int_0^1 \wb(x)^T \bmat{B_1(x) \\ B_2(x)}dx \cdot  Fw_\alpha^b \\
=& \int_0^1 \wb(x)^T \bmat{B_1(x) \\ B_2(x)}  Fw_\alpha^b dx \\
=&  \int_0^1 \wb(x)^T \bmat{B_1(x) \\ B_2(x)}  \bmat{0_{\beta\alpha,\beta(\alpha+1)} & F}\wb(x) dx \\
=&   \int_0^1 \wb(x)^T \bmat{B_1(x) \\ B_2(x)} \bmat{0_{\beta(\alpha+1)} & B_1(x)F \\ 0_{2\beta\alpha,\beta(\alpha+1)} & B_2(x)F} \wb(x)dx \\
=&   \int_0^1 \wb(x)^T B_b(x) \wb(x)dx.
\end{align*}   Then, writing the kernel in symmetric form completes the proof.
\end{proof}

%%%%%%%%%%%%%%%%%%%%%%%%%%%%%%%%%%%%%%%%%%%%%%%%%%%%%%%%%%%%%%%%%%%%%%%%%%%

Finally, we provide the following result which we will use in the proof of Theorem~\ref{thm:non_neg} and Corollary~\ref{cor:pos}.

\begin{lemma}\label{lem:order}
For any $v \in \mcl{L}_2\left([0,1];\R^\beta \right)$ and polynomial matrices $F_1,F_2,G_1,G_2 \in  \mcl{R}^{\beta \times \beta}[x,y]$, the following identity holds
\begin{align}
&\int_0^1 \left( \int_0^x F_1(x,y)v(y)dy + \int_x^1 F_2(x,y)v(y)dy  \right)^T \notag \\
& \times \left( \int_0^x G_1(x,y)v(y)dy + \int_x^1 G_2(x,y)v(y)dy  \right)dx \notag \\
&\label{eqn:order} = \hlf \int_\Delta v(x)^T \Gamma \left[K \right] v(y) d\Delta,
\end{align} where
\begin{align*}
&K(x,y)\\
=& \int_0^y \left(F_2(z,x)^T G_2(z,y) + G_2(z,x)^T F_2(z,y)   \right)dz \\
& + \int_y^x \left(F_2(z,x)^T G_1(z,y) + G_2(z,x)^T F_1(z,y)   \right)dz \\
& + \int_x^1 \left(F_1(z,x)^T G_1(z,y) + G_1(z,x)^T F_1(z,y)  \right)dz.
\end{align*} The same result holds for any $v \in \mcl{H}^\alpha\left([0,1];\R^\beta\right)$, $F_1,F_2,G_1,G_2 \in \mcl{R}^{\beta(\alpha+1) \times \beta(\alpha+1)}[x,y]$ and $v(x)$ replaced by $v_\alpha(x)$ in~\eqref{eqn:order}.
\end{lemma}
%%%%%%%%%%%%
\begin{proof}
We begin by observing that the left hand side of~\eqref{eqn:order} may be written as
\begin{equation}
\ip{\mcl{F} v}{\mcl{G} v}_{\mcl{L}_2}=\ip{ v}{\mcl{F}^\star \mcl{G} v}_{\mcl{L}_2},
\end{equation}
where the linear bounded operators on $\mcl{L}_2\left([0,1];\R^\beta\right)$ are defined as
\begin{align*}
\left(\mcl{F}  v  \right)(x)=& \int_0^x F_1(x,y)v(y)dy + \int_x^1 F_2(x,y)v(y)dy,\\
\left(\mcl{G}  v  \right)(x)=& \int_0^x G_1(x,y)v(y)dy + \int_x^1 G_2(x,y)v(y)dy,
\end{align*} and where the adjoint operator is given by
\[\left(\mcl{F}^\star  v  \right)(x)= \int_0^x F_2(y,x)^T v(y)dy + \int_x^1 F_1(y,x)^T v(y)dy.\]
Now, using the fact that
\[\left(\mcl{G}  v  \right)(y)= \int_0^y G_1(y,z)v(z)dz + \int_x^1 G_2(y,z)v(z)dz,\] we get
\begin{align}
\left( \mcl{F}^\star \mcl{G} v \right)(x)=& \int_0^x F_2(y,x)^T \left(\mcl{G} v \right)(y)dy \notag \\
& + \int_x^1 F_1(y,x)^T \left(\mcl{G} v \right)(y)dy \notag \\
=&\label{eqn:order1}\sum_{i=1}^4 \Theta_i(x),
\end{align} where
\begin{align*}
\Theta_1(x)=& \int_0^x \int_0^y F_2(y,x)^T G_1(y,z)v(z)dz dy,\\
\Theta_2(x)=& \int_0^x \int_y^1 F_2(y,x)^T G_2(y,z)v(z)dz dy,\\
\Theta_3(x)=& \int_x^1 \int_0^y F_1(y,x)^T G_1(y,z)v(z)dz dy,\\
\Theta_4(x)=& \int_x^1 \int_y^1 F_1(y,x)^T G_2(y,z)v(z)dz dy.
\end{align*} In each of the $\Theta_i(x)$ we change the order of integration and switch between the variables $y$ and $z$ to obtain
\begin{align*}
\Theta_1(x)=& \int_0^x \int_y^x F_2(z,x)^T G_1(z,y)dz v(y)dy,\\
\Theta_2(x)=& \int_0^x \int_0^y F_2(z,x)^T G_2(z,y)dz v(y)dz \\
&+ \int_x^1 \int_0^x F_2(z,x)^T G_2(z,y)dz v(y)dy,\\
\Theta_3(x)=& \int_0^x \int_x^1 F_1(z,x)^T G_1(z,y)dz v(y)dy \\
& + \int_x^1 \int_y^1 F_1(z,x)^T G_1(z,y)dz v(y)dy,\\
\Theta_4(x)=& \int_x^1 \int_x^y F_1(z,x)^T G_2(z,y)dz v(y)dz.
\end{align*} Substituting into~\eqref{eqn:order1} and consequently in~\eqref{eqn:order} we get
\begin{align}
\ip{\mcl{F}v}{\mcl{G}v}_{\mcl{L}_2}=&\int_0^1 \int_0^x v(x) K_1(x,y)v(y)dy dx \notag \\
&\label{eqn:order2}+\int_0^1 \int_x^1 v(x) K_2(x,y)v(y)dy dx,
\end{align} where
\begin{align*}
K_1(x,y)=& \int_0^y F_2(z,x)^TG_2(z,y)dz \\
&+ \int_y^x F_2(z,x)^T G_1(z,y)dz \\
& + \int_x^1 F_1(z,x)^T G_1(z,y)dz,\\
K_2(x,y)=& \int_0^x F_2(z,x)^T G_2(z,y)dz \\
&+ \int_x^y F_1(z,x)^T G_2(z,y)dz \\
&+ \int_y^1 F_1(z,x)^T G_1(z,y)dz.
\end{align*} Then, applying Lemma~\ref{lem:switch} to~\eqref{eqn:order2} completes the proof.
\end{proof}

\bibliographystyle{plain}
\bibliography{CDC17}

\end{document}